\documentclass[conference,twocolumn,10pt]{IEEEtran}
\usepackage{mathrsfs}

\usepackage{cite}

\usepackage{graphicx}

\usepackage{algorithm}

\usepackage{psfrag}

\usepackage{amsthm,amssymb}

\usepackage{color}

\usepackage{subfigure}

\usepackage{url}

\usepackage{stfloats}  

\usepackage{amsmath}   
\usepackage{chngpage}

\usepackage{array}

\newcommand{\picspace}{{\vspace{-0.1 in}}}          

\begin{document}

\title{Deploying Multiple Antennas on High-speed Trains: Equidistant Strategy v.s. Fixed-Interval  Strategy}

\author{Yang Lu$^{\dagger}$, Ke Xiong$^{\dagger,\natural}$, Pingyi Fan$^{\ast}$,~Zhangdui Zhong$^{\dagger}$\\
\small
$^\dag$School of Computer and Information Technology, Beijing Jiaotong University, Beijing, P. R. China\\
$^\natural$National Mobile Communications Research Laboratory, Southeast University, Nanjing, P. R.China\\
$^\ast$Department of Electronic Engineering, Tsinghua University, Beijing, P. R. China\\}

\maketitle

\begin{abstract}
Deploying multiple antennas on high speed trains is an effective way to enhance the information transmission performance for high speed railway (HSR) wireless communication systems. However, how to efficiently deploy multiple antennas on a train? This problem has not been studied yet. In this paper, we shall investigate efficient antenna deployment strategies for HSR communication systems where two multi-antenna deployment strategies, i.e., the equidistant strategy and the fixed-interval strategy, are considered.  To evaluate the system performance,  mobile service amount and outage time ratio are introduced. Theoretical analysis and  numerical results show that, when the length of the train is not very large, for two-antenna case, by increasing the distance of neighboring antennas
in a reasonable region, the system performance can be enhanced.
It is also shown that the two strategies have much difference performance behavior in terms of instantaneous channel capacity,
 and the fixed-interval strategy may achieve much better performance than the equidistant one in terms of service amount and outage time ratio when the antenna number is much large.
\end{abstract}

\section{Introduction}

High speed railway (HSR) system is developed rapidly all over the world. Especially in China, the length of railway has exceeded 11028km. Since it is a fast, safe and green transportation system, HSR is playing a very important role in people's daily life, and more and more people are spending more and more time on high speed trains (HST) for traveling. In order for passengers to enjoy the information services including entertainment and online video, the demand of wireless communications in trains grows day by day. But due to the high mobility of trains, this communication demand is challengeable to be satisfied.

To provide reliable communications for HSR systems, several problems have to be solved \cite{lter}. Firstly, the handover between the cells is very frequently because of high moving speed. For instances, the cell radius is generally about 2$\sim$3km in typical GSM-R networks \cite{GSR}. If the speed of train is over 500km/h \cite{speed}, the handover period is about 14$\sim$21s, which is quiet short  compared with the general information services. Secondly, the Doppler shift makes the synchronization quite difficult. For example, the maximum Doppler shift of HST is about 945Hz at the speed of 486km/h with the carrier frequency of 2.1 GHz. Thirdly, the HST carriage is enclosed by metal materials which shields the signals transmitted in and out the carriage, termed as penetration loss. In practical HSR systems, the penetration loss even exceeds 25dB. To release these problems above, a great deal of work so far has been done, e.g., \cite{handover}-\cite{two-top}. In \cite{handover}, it proposed a seamless dual-link handover scheme to minimize the communication interruption during handover to handle the first problem. In \cite{Doppler}, a cooperative antenna system was built and an algorithm of Doppler frequency offset (DFO) estimation was proposed. In \cite{two-top}, a two-hop architecture was proposed to overcome the dramatic penetration loss of radio signals transmitted in and out the HST carriages, where all users in the train are considered as a big virtual user, their transmitted signals are first aggregated at a mobile relay deployed on the carriage and then  delivered over the train-ground link between the mobile relay and the ground BSs. Besides, the deployment of BS was investigated in \cite{group} and \cite{service}.

In the past decades, multi-antenna technology has been widely investigated, see e.g., \cite{service},\cite{multi1}, which is able to provide diversity and multiplexing gain to improve the information transmission throughput and reliability of wireless communications. Deploying multiple antennas requires that the antenna spacing must be larger than 1/2 wavelength so that the transmit/receive signals for the multiple antennas have relatively low correlation.  Since a HST is usually with the length about 200m, it provides sufficient space for deploying multiple antennas in linear array. Recently, some works started to study  the multi-antenna system deployed on the train \cite{handover}.

 The deployment of multiple antennas may impact the system performance of HSR communications greatly. As reported in HSR system, the large scale fading of the channels is the dominant factor affecting the system communication performance. When the train is moving, the distances between the multiple antennas deployed on the train and the serving BS located on the railway side change very fast. If the antennas are equipped very closely to each other, they may have very similar distances to the serving BS. While if the antennas are positioned relatively far from each other, they may have very different distances to the serving BS. Thus, different antenna deployment strategies may result in very different performance. To better understand the impacts of the deployment of multiple antennas on HST, several fundamental questions are required to be answered.
\begin{itemize}
  \item How does the deployment of multiple antennas impacts the instantaneous channel capacity and the total service amount of HST?
  \item Besides the instantaneous channel capacity and service amount, how to effectively evaluate the system information transmission performance over the time domain?
  \item What is the efficient multiple-antenna deployment strategy?
\end{itemize}

To answer these questions, we consider two multiple-antenna strategies i.e., the equidistant strategy and the fixed-interval strategy.  Besides the instantaneous channel capacity and service amount, outage time ratio is also introduced as a metric to evaluate the system performance.
  Some theoretical analysis and numerical results show that, for two-antenna case, by increasing the distance of neighboring antennas in a reasonable region, the system performance can be enhanced and  that the
  two strategies have much difference performance behavior in terms of instantaneous channel capacity
   and the fixed-interval strategy  may achieve better performance than the equidistant one in terms of service amount and outage time ratio if the antenna number is relatively large.

The rest of this paper is organized as follows: In Section II, it introduces the system model and two multiple-antenna deployment strategies. In Section III, we define the performance metrics and analyze the two multiple antennas deploying strategies. Some simulation results are provided in Section IV. Finally,  some conclusions are given in Section V.

\section{System Model}

\subsection{Network Model}
Consider a HSR communication scenario as shown in Figure  \ref{systemmodel}, where a HST runs at a constant speed $v$ along a straight track\footnote{For a small time period, the moving speed of the train is usually considered to be stable and unchanged.} and it desires to communicate  with the serving base station (BS) via a two-hop link. The two-hop link means that there is a relay equipped on the train, which is used to help the information delivering between the serving BS and the users in the train. Specifically, in the downlink scenario, the roadside BS first transmits information to the relay and then the relay helps to forward the information to the passengers within the carriages. In the uplink scenario, the passengers in the carriages first connect with the relay and then transmit information to the BS with the help of the relay. Since the service amount of downlink is much larger than that of uplink, in this paper, we only focus on the downlink transmission over the two-hop link.

\subsection{Channel Model}
Suppose the BSs are positioned with equal distance interval $D$ on the parallel line of the railway.
To enhance the system transmission performance, $N$ omnidirectional antennas\footnote{It is known that when the distance between two antennas exceeds half a wavelength, these two antennas can be regarded as irrelevant. Thus, there is a limit of $N$. For example, in LTE system, the carrier frequency is about 2000MHz and the wavelength is about 0.15m. So, the distance between two antennas must be larger than $7.5\times 10^{-2}$m and the maximal $N$ is about 2666.} are equipped on the train for the relay, referred to as train relay station (TRS). Therefore, in the downlink transmission, TRS can receive multiple copies of the signals associated with same information transmitted from the serving BS. Then, the TSR can decode the information from the  BS by using Maximal-Ratio-Combining (MRC). For clarity, we use $B$ and $n$ to represent the BS and the $n-$th antenna respectively, where $1 \le n \le N$.

\begin{figure}
\centering
\includegraphics[width=0.45\textwidth]{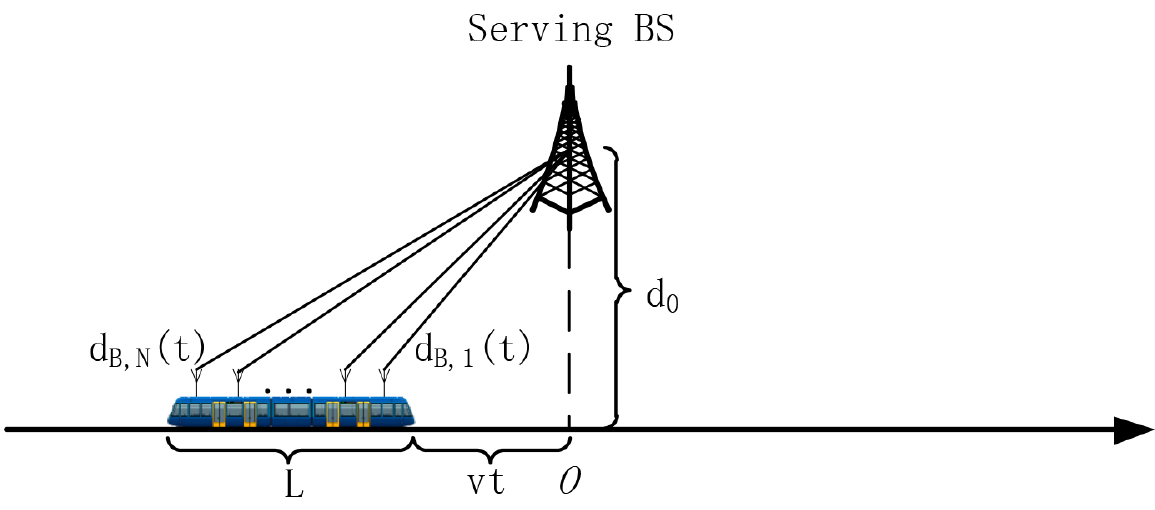}
\caption{System Model}
\label{systemmodel}
\end{figure}\picspace

In HSR communications, the effect of large-scale fading is more important comparing with that of small-scale fading, since in most scenarios, HST is moving in open fields such as the viaduct bridge, which is lack of scatters. In this case, the distance between the HST and the BS {is the dominant factor affecting} the received signals. The path loss can be expressed as $ {P_L}\left(d_{B,n}\right) =  - 10\log \left( {{G_l}{\lambda ^2}/{{\left(4\pi {d_{B,n}}\right)}^2}} \right)$, where $G_l$ is the antenna gain, $\lambda$ is the wavelength of the carrier and ${d_{B,n}}$ is the distance between $n-$th ($n=1,...,N$) antenna of HST and the BS.

\subsection{Two Multiple-Antenna Strategies}
Since the length of a train is often more than 200 meters, it provides a large space to deploy multiple antennas in linear mode and to obtain special diversity. As mentioned previously, large-scale fading is the dominant factor for the HSR communications, which indicates that if one places the antennas at different position, it may have big different large-scale effects, resulting in much difference in system performance. Motivated by this, we present two antenna deployment strategies as follows.

\subsubsection{Equidistant Strategy}

In equidistant deployment strategy, the antennas are positioned evenly on the train with equal distance between any two neighboring antennas, where the equal distance should be large than the half wavelength and set to be $\frac{L}{N-1}$m, where $L$ is the length of the HST,  $N$ is the antenna number deployed on the train.  For example, suppose $L=200$m, if two antennas are deployed according to the equidistant strategy, the distance between the two antennas is 200m, while for three antennas case, the distance between two neighboring antennas is 100m.

Assuming the cross point of the railway and its vertical line through the serving BS is the  original point of the time axis and train moves towards the positive direction of the time axis. The distance between $n-$th antenna and the serving BS at time $t$ is

\begin{equation}
d_{B,n}^{\left( {Equ} \right)}\left( t \right) = \sqrt {{{\left( {vt - \tfrac{L}{{N - 1}}\left( {n - 1} \right)} \right)}^2} + d_0^2}
\end{equation}
where  $v$ is the velocity of HST, $d_0$ is the vertical distance between the track and the serving BS and $t \in \left[ - D/2v,D/2v\right]$, $D$ is the serving coverage parameter along the high speed train track.

Additive white Gaussian noise (AWGN) channel {is assumed, so} the $n-$th antenna's received signal to noise ratio {(SNR) from the BS can be expressed by} $SNR{_{B,n}^{\left( {Equ} \right)}}\left( t \right)={P_t} {P_L}({d_{B,n}^{\left( {Equ} \right)}\left( t \right)})/{N_0}W$. {Therefore, the} instantaneous {achievable information rate per Hz} of the HST with equidistant strategy at time $t$ is
\begin{equation}\label{c_t}
C^{\left( {Equ} \right)}\left( t \right) = \log \left( {1 + \sum\nolimits_{n = 1}^N {SNR{_{B,n}^{\left( {Equ} \right)}}} } \right)
\end{equation}
where  $W$ is the bandwidth of carrier, $ {P_L}({d_{B,n}^{\left( {Equ} \right)}\left( t \right)})$ is the path loss from $n-$th antenna to the serving BS, $N_0$ is the power density of AWGN and $P_t$ is the transmission power of the serving BS.

\subsubsection{Fixed-interval Strategy}

In the fixed-interval deployment strategy, the antennas are divided into two sets with equal amount and the antennas in the first set is placed one by one from the head to the tail of the train while the antennas in the second set is placed one by one from the tail to the head of the train. The distance between any two neighboring antennas in the same set is with a fixed value $\delta$, where $\delta$ is larger than the half wavelength and less than $\frac{L}{N-1}$ m.

Assuming that the $N$ is even, the distance between $n-$th antenna and the serving BS can be given by
\begin{equation}
d_{B,n}^{\left( {Fix} \right)}\left( t \right) = \left\{ \begin{array}{l}
\sqrt {{{\left( {vt - n\delta } \right)}^2} + d_0^2} \left( {n \le \frac{N}{2}} \right)\\
\sqrt {{{\left( {vt - \left( {L - n\delta } \right)} \right)}^2} + d_0^2} \left( {n > \frac{N}{2}} \right)
\end{array} \right.
\end{equation}

The corresponding received $SNR{_{B,n}^{\left( {Fix} \right)}}\left( t \right)={P_t} {P_L}({d_{B,n}^{\left( {Fix} \right)}\left( t \right)})/{N_0}W$ and the instantaneous achievable information rate of HST with fixed-interval strategy at time t is

\begin{equation}\label{c_t}
C^{\left( {Fix} \right)}\left( t \right) = \log \left( {1 + \sum\nolimits_{n = 1}^N {SNR{_{B,n}^{\left( {Fix} \right)}}} } \right)
\end{equation}

\section{PERFORMANCE ANALYSIS}
\subsection{Performance Metrics}
To effectively evaluate the performance of different antenna deployments, besides instantaneous achievable information rate, we introduce the concepts of mobile service\cite{service} and outage time ratio of the HST during the time period $[ - D/2v,D/2v]$.

\subsubsection{Mobile Service}
Service was defined as the integral of the instantaneous channel capacity over the time period $\left[ - D/2v,D/2v\right]$, which is used to describe the information transmission amount of the serving BS during a given time period and can be considered as the upper bounds of transmission\cite{service}.
\begin{equation}\label{C_sum}
{S} = \int_{ - D/2v}^{D/2v} {C\left(t\right)} dt
\end{equation}
where $C\left(t\right)$ is $C^{\left(Equ\right)}\left(t\right)$ or $C^{\left(Fix\right)}\left(t\right)$.

\subsubsection{Outage Time Ratio}

In practical applications, when the instantaneous channel capacity is lower than the user's information rate requirement, the communication outage occurs. Thus,
for a given rate threshold ${C_{\rm th}}$, we define the total time duration that the
the instantaneous channel capacity lower than ${C_{\rm th}}$ to the total time period $T$ as the outage time ratio, which can be calculated by
\begin{equation}
\begin{aligned}
{P_{\rm OTR}} &= \frac{\sum_i(T^{+}_i-T^{-}_i)}{T}
\end{aligned}
\end{equation}
where $T=D/2v-\left(-D/2v\right)=D/v$ and $C\left(t\right)<{C_{\rm th}}$ for $t\in [T^{-}_i,T^{+}_i]$ and $T^{-}_i\geq 0$ and $T^{+}_i\leq T$.

\begin{figure}
\centering
\includegraphics[width=0.495\textwidth]{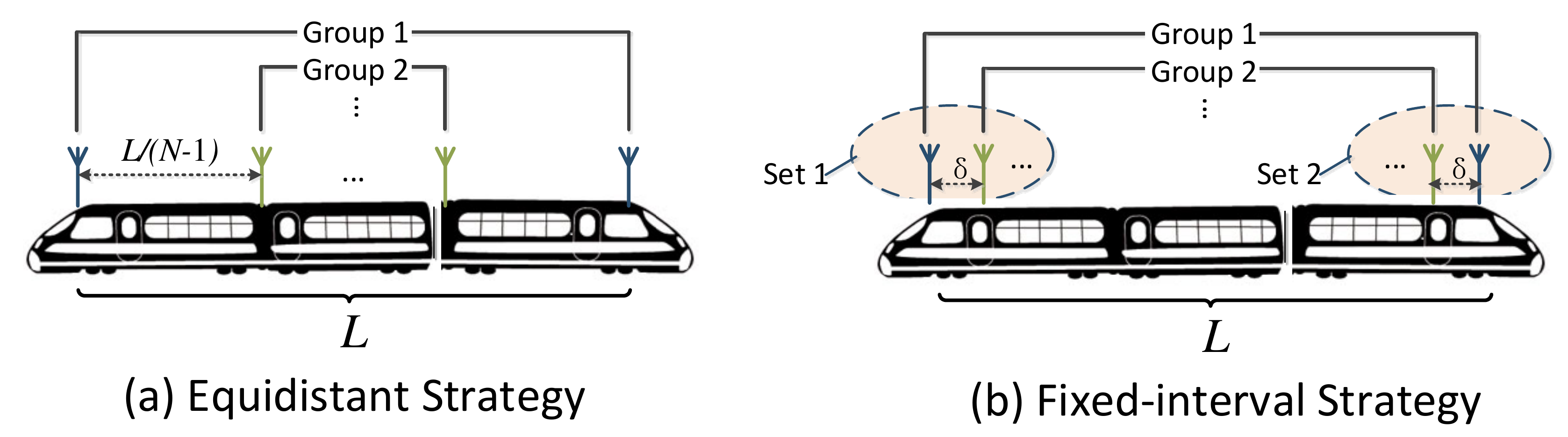}
\caption{The comparison of two multiple-antenna systems: (a) Equidistant strategy; (b) Fixed-interval strategy.}
\label{four_antenna}
\end{figure}\picspace

\subsection{Analysis}
In this subsection, we shall analyze the performance of the two proposed strategies. As it is difficult to directly calculate service amount and outage time ratio for the two strategies with a relatively large $N$, we analyze them as follows.

We first analyze the case of $N=2$.

It is noted that if the fixed interval is equal to the train length $L$, i.e. the two antenna are deployed at the head and tail of the train,
 the fixed-interval strategy is the same with that of the equidistant strategy with distance between two antennas as $L$. That is,  both strategies have the same system performance.
 If the fixed interval is less than the train length $L$, and the distance between two antennas is not the same,  both of the strategies will have different performance.

For $N=2$,  consider the case that the two antenna are deployed at the head and tail of the train, the mobile service can be calculated as Eq.(\ref{C_sum}) \cite{service}. Let $f\left(t,L\right)=1 + SN{R_{B,1}} + SN{R_{B,2}}$. Because $f\left(t,L\right)$ is always $>1$,  $\log\left(f\left(t,L\right)\right)$ is always larger than $0$.  The relation between ${S}$ for two antennas and $L$ can be summarized as follows.
\newtheorem{lemma}{\textbf{Lemma}}
\begin{lemma}
For given $D$ and $v$, $S_{two}$ for two antennas is decreasing function w.r.t. $L$  if $L$  is very large and it has a maximum.
\end{lemma}

\begin{proof}
Let $SNR_0 = {{{P_t}{G_l}{\lambda ^2}}}/{{{{\left( {4\pi } \right)}^2}{N_0}B}}$,
\begin{equation}\label{proof1}
\begin{aligned}
\frac{{d{S_{\rm two}}}}{{dL}} &= \frac{{d\int_{ - D/v}^{D/v} {\log \left( {f\left( {t,L} \right)} \right)} dt}}{{dL}}\\
 &= \int_{ - D/v}^{D/v} {\frac{1}{{f\left( {t,L} \right)\ln 2}}\frac{{df\left( {t,L} \right)}}{{dL}}} dt.
\end{aligned}
\end{equation}
Since $f\left(t,L\right)>1$ always holds and
\begin{equation}
\frac{{df\left( {t,L} \right)}}{{dL}} = \frac{{SN{R_0}\left( {2\left( {vt - L} \right)} \right)}}{{{{\left( {{{\left( {L - vt} \right)}^2} + d_0^2} \right)}^2}}},
\end{equation}
it can be inferred that $\left(vt-L\right)<0$ when $L>D$ for each $t\in [-D/v, D/v]$. In this case, the integral function is always negative, which imply that $\frac{{d{S_{\rm two}}}}{{dL}}<0$.  Thus $S_{two}$ for two antennas is decreasing function in $L$.
 On the other hand, $S_{two}$ is a continuous function of $L$ and bounded,  thus  there exists a maximum.
\end{proof}

\begin{lemma}
For given $D$ and $v$,  $P_{\rm OTR}$ with $N=2$ is a first decrease  and then increase function w.r.t. $L$ and

\begin{equation}\label{pout}
\begin{aligned}
{P_{\rm OTR}} = 1-\tfrac{{\sqrt {\frac{{\left( {\beta {L^2} - 4\beta d_0^2 + 4\sqrt {\left( { - {\beta ^2}{L^2}d_0^2 + \beta {L^2} + 1} \right)}  + 4} \right)}}{\beta }} }}{D},
\end{aligned}
\end{equation}
where $\beta  = \frac{{{e^{C_{\rm th}}} - 1}}{{SNR_0}}$.
\end{lemma}

\begin{proof}
For $n=2$, the differential of $C(t)$ w.r.t. $t$ is
\begin{equation}
\begin{aligned}
\frac{{d{C(t)}}}{{dt}} = \frac{{\frac{{2SN{R_0}v(L - vt)}}{{{{\left( {{{\left( {L - vt} \right)}^2} + d_0^2} \right)}^2}}} - \frac{{2SN{R_0}t{v^2}}}{{{{\left( {d_0^2 + {{\left( {vt} \right)}^2}} \right)}^2}}}}}{{\left( {\log \left( {\frac{{SN{R_0}}}{{d_0^2 + {{\left( {vt} \right)}^2}}} + \frac{{SN{R_0}}}{{{{\left( {L - vt} \right)}^2} + d_0^2}}} \right) + 1} \right)}}
\end{aligned}
\end{equation}
Since the denominator is always $>0$, the root of $\frac{{d{C(t)}}}{{dt}} = 0$ is

\begin{equation} \label{t3}
\left\{
\begin{aligned}
{t_1} &= \tfrac{{\sqrt {L - \sqrt {2L({L^2} + 4d_0^2)}  - {L^2} - 4d_0^2} }}{{2v}},\\
{t_2} &= \tfrac{L}{{2v}},\\
{t_3} &= \tfrac{{\sqrt {L + \sqrt {2L({L^2} + 4d_0^2)}  - {L^2} - 4d_0^2} }}{{2v}},
\end{aligned} \right.
\end{equation}
where ${t_2}$ is a repeated root representing the local minimum point and $t_1$ and $t_3$ represent the rate peak point.
Let $\beta  = \frac{{{e^{C_{\rm th}}} - 1}}{{SNR_0}}$. The solutions of $C(t)={C_{\rm th}}$ w.r.t. $t$ are as follows. They are,
${T_1} = \tfrac{L + \Phi}{{2v}}$,
${T_2}= \tfrac{L - \Phi}{{2v}}$,
${T_3}= \tfrac{L - \Psi}{{2v}}$,
${T_4}= \tfrac{L + \Psi}{{2v}}$,
where $\Phi=\sqrt {\tfrac{{\left( {Q + \Theta} \right)}}{\beta }}$, $\Psi=\sqrt {\tfrac{{\left( {Q -\Theta} \right)}}{\beta }}$, $Q=\beta ({L^2} - 4d_0^2)$ and $\Theta=4\sqrt {\left(\beta {L^2}- {\beta ^2}{L^2}d_0^2 + 1 \right)} + 4$.
It's worth noting that within $\left[ { - D/2v,D/2v} \right]$, ${C_{\rm th}}$ has at most four intersect points with the curve of capacity of the two-antenna system. Particularly, when ${{C_{\rm th}}} \in \left( {C\left( {{t_1}} \right), + \infty } \right)$, there is no intersection point. When ${{C_{\rm th}}} \in \left( {C\left( {{t_2}} \right),C\left( {t_1} \right)} \right)$, there are four intersection points, i.e., $T_1$, $T_2$, $T_3$ and $T_4$. When ${{C_{\rm th}}} \in \left( {0,C\left( {{t_2}} \right)} \right)$, there are two intersection points, $T_1$ and $T_2$, and in this case, the outage time ratio is ${P_{\rm OTR}} =1- \frac{{{T_1} - {T_2}}}{T}$. Therefore,
$
{P_{\rm OTR}} = 1-{\sqrt {\tfrac{{\left( {Q + \Theta} \right)}}{\beta }} }/{D}.
$
With some simple mathematical manipulations, one can arrive at Lemma 2.
\end{proof}


\begin{figure*}[!ht]
\begin{equation}\label{int}
\begin{aligned}
S_{\textrm{two}}= \int\limits_{ - D/2v}^{D/2v} {\log \left( {1 + SN{R_{B,1}} + SN{R_{B,2}}} \right)} dt = \int\limits_{ - D/2v}^{D/2v} {\log } \left( {1 + \tfrac{{{P_t}{G_l}{\lambda ^2}}}{{{{\left( {4\pi \sqrt {{{\left( {vt} \right)}^2} + {d_0}^2} } \right)}^2}{N_0}B}} + {\rm{ }}\tfrac{{{P_t}{G_l}{\lambda ^2}}}{{{{\left( {4\pi \sqrt {{{\left( {vt - L} \right)}^2} + {d_0}^2} } \right)}^2}{N_0}B}}} \right)dt
\end{aligned}
\end{equation}
\hrule
\end{figure*}

For $N>2$, assuming that $N$ is an even number, then we have a lemma as follows.
\begin{lemma}
For $N>2$, when $L$ is not very large, the service amount of the fixed-interval strategy is not lower than that of the equidistant strategy and the outage time ratio of the fixed-interval strategy is not higher than that of the equidistant strategy.
\end{lemma}

\begin{proof}
All $N$ antennas can be divided into $N/2$ groups and each group includes two antennas which are expressed as the $n$-th antenna and the $(N+1-n)-$th antenna, where $n=1,...,N/2$. For example, the first antenna and the last one is in group 1, and the second antenna and the $N-1$-th antenna is in group 2, as shown in Figure \ref{four_antenna}.
With the grouping method of the antennas, for the group with the same index $m$ ($m=1,2,...,N/2$), the distance of between two antennas in equidistant strategy is
$
d_m^{Equ}=L-2\frac{L}{{N - 1}}\left(m-1\right)
$
and that in fixed-interval strategy is
$
d_m^{Fix}=L-2\delta\left(m-1\right).
$
Because $\frac{L}{{N - 1}} \ge \delta$,  $d_m^{Equ} \le d_m^{Fix}$, we have that the distance between two antennas in each group of the fixed-interval strategy is no less than that in the corresponding group of the  equidistant strategy.  With lemma 1 and lemma 2, when $L$ is not large, for each corresponding group, it can be inferred that the service amount and the  $P_{\rm OTR}$ in fixed-interval strategy are respectively higher and lower than those in the equal-distant strategy. Therefore, Lemma 3 is proved.
\end{proof}


\newtheorem{remark}{\textbf{Remark}}
 \begin{remark}
When $N=N_{\textrm{max}}$, the deployment result of the fixed-interval strategy is the same with that of the equidistant strategy, where ${N_{{\rm{max}}}} = \frac{L}{\delta } + 1\left( \delta  \ge  \frac{{{L_{wave}}}}{2}\right)$, $L_{\textrm{wave}}$ is the wavelength of carrier. That is, when $N=N_{\textrm{max}}$, both strategies have the same system performance.
\end{remark}

\section{RESULTS AND DISCUSSION}
In this section, we provide some simulation results to discuss the system performance of the proposed two antenna deployment strategies in terms of instantaneous capacity, service mount and outage time ratio. Detailed simulation parameters are shown in Table \ref{parameter}.

\begin{small}
\begin{table}[h]\normalsize
\centering
\caption{Simulation parameter}
\begin{tabular}{  l c }
\hline
Parameter & Value \\
\hline
Train moving Speed & 100m/s \\
Carrier frequency & 2GHz \\
Bandwidth of each subcarrier & 15kHz \\
Noise density & 145dBm/Hz \\
Maximum received SNR & 5dB \\
Minimum distance ($d_0$) & 50m \\
Length of train ($L$) & 200m\\
Time period & $\left[-6,6\right]$\\
\hline
\label{parameter}
\end{tabular}
\end{table}
\end{small}

Figure \ref{d_L_t} shows the system performance versus the distance between two neighboring antennas for $N=2$ case, where Figure \ref{d_L_t}(a) and (b) plot the service amount and outage time ratio verses antenna distance for equidistant strategy, respectively. The antenna distance is change from $0.1m$ to $200m$ for the equidistant strategy and the fix interval $\delta$ is set to be 1m for the fixed interval strategy. $C_{\textrm{th}}$ is set as 0.15 bit/s/Hz which is the average capacity of single antenna system during the given time period.  It can be observed that with the increment of the antenna distance, the service increases and the outage time ratio decreases. This is because the derivative of function $\log(x)$ declines as $x$ increasing. Particularly, when the two antennas are deployed close to each other, they pass the original point almost at the same time. As a result, $SNR_{B,1}$ and $SNR_{B,2}$ achieves their maximum values almost at the same time. In this case, $\sum\nolimits_n {SN{R_{B,n}}}$ contribute relatively little to improve total service amount. When the two antennas are deployed very far from each other, they achieve their maximum SNR at very different time. Thus, the integration of instantaneous channel capacity is improved greatly. Moreover, as the two antennas achieve their maximum received SNR at different time, thus the outage time ratio is decreased.

Figure \ref{L1000} plots an extreme case with two antennas, where the length of the train is sufficiently large, so that the distance of antennas  can be increased much larger. It is observed that the service amount firstly increases and then decrease with the increment of antenna distance, which is consist the theoretical results obtained in Lemma 1.
\begin{figure}
\centering
\includegraphics[width=0.5\textwidth]{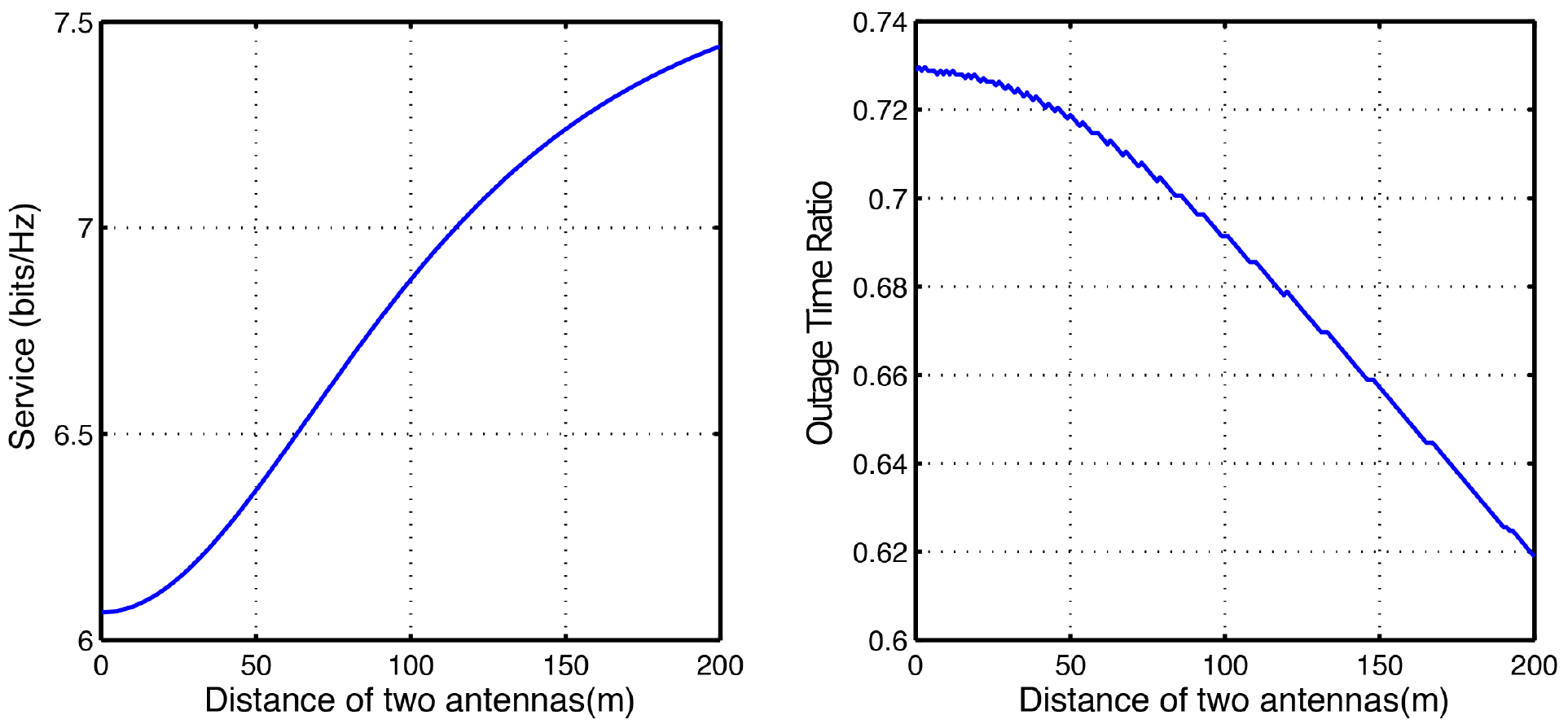}\picspace
\caption{The service and outage time ratio versus the distance between two antennas for $N=2$ case.}
\label{d_L_t}
\end{figure}\picspace

\begin{figure}
\centering
\includegraphics[width=0.25\textwidth]{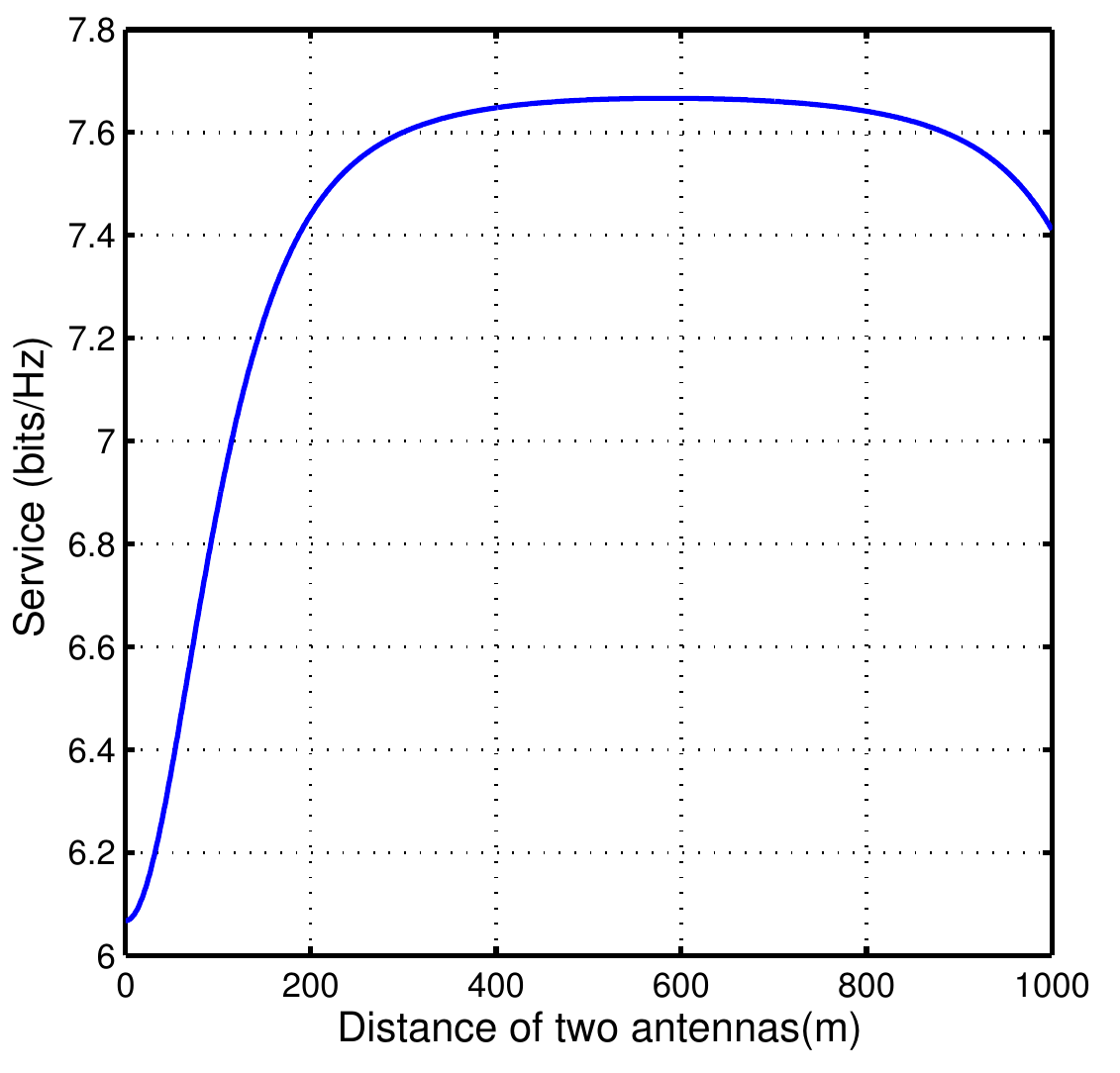}\picspace
\caption{The service amount versus the distance between two antennas for $N=2$ case.}
\label{L1000}
\end{figure}\picspace

Figure \ref{fix} (a) and (b) show the instantaneous channel capacity of equidistant strategy and fixed-interval strategy where $\delta$ is set as 1m with the number of antennas being 10, 50, 100, 200, respectively. It is observed that, with the increment of the number of antennas, the instantaneous channel capacity of both two strategies can be increased at any time. But the capacity gain is relatively large when the train is near the serving BS and relatively small when the train is far way from the serving BS. Besides, the two strategies show very different instantaneous channel capacity behavior with the increment of the number of antennas. That is, there are two rate peaks in the fixed-interval strategy, but basically only one peak in the equidistant strategy when the number of antennas is 10, 50 and 100. It is also shown that in fixed-interval strategy, the distance between two rate peaks decreases and finally the two rate peaks are merged into one rate peak when the number of antennas is relatively large. Moreover, when $N=200$, the two strategies show very similar performance behavior ,since in this case the deployment of two strategies are the same.

\begin{figure}
\centering
\includegraphics[width=0.5\textwidth]{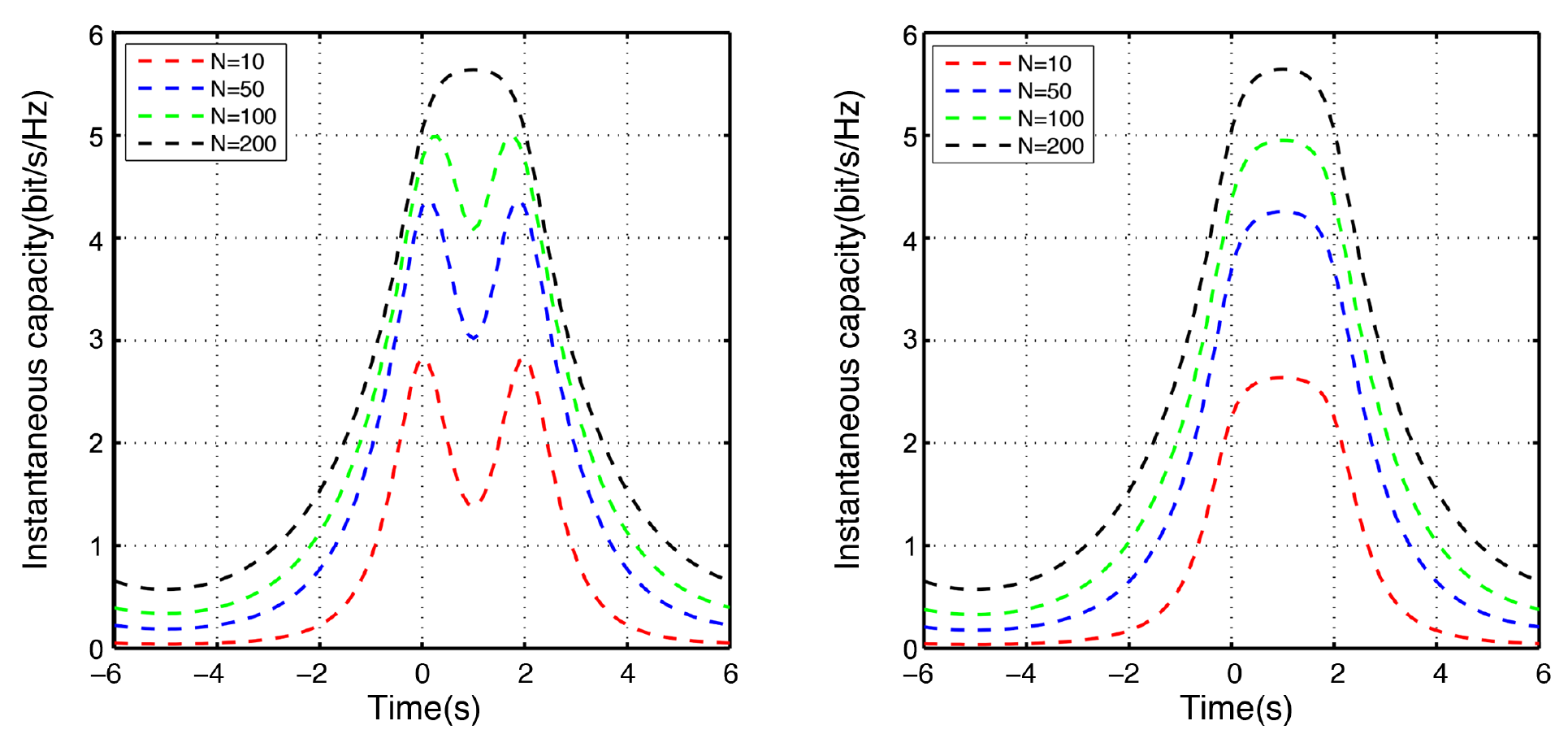}\picspace
\caption{The instantaneous channel capacity of (i) fixed-interval strategy ($\delta=1m$) and (ii) equidistant strategy.}
\label{fix}
\end{figure}

In order to fully compare the instantaneous channel capacity of the two strategies, Figure \ref{A600} changes $\delta$ to be 0.15m, which is similar to the wavelength of the carrier. In this case, 600 antennas are deployed on the HST, which can be considered as a massive antenna case. It can be observed that the fixed-interval strategy contributes a larger capacity gain than the equidistant strategy in most time, but it has lower peak rate than the equidistant one. As it is more expected to experience a more smooth wireless communication rather than very high peak rate, it is preferred to adopt the fixed-interval strategy in such a massive antenna case.

\begin{figure}
\centering
\includegraphics[width=0.25\textwidth]{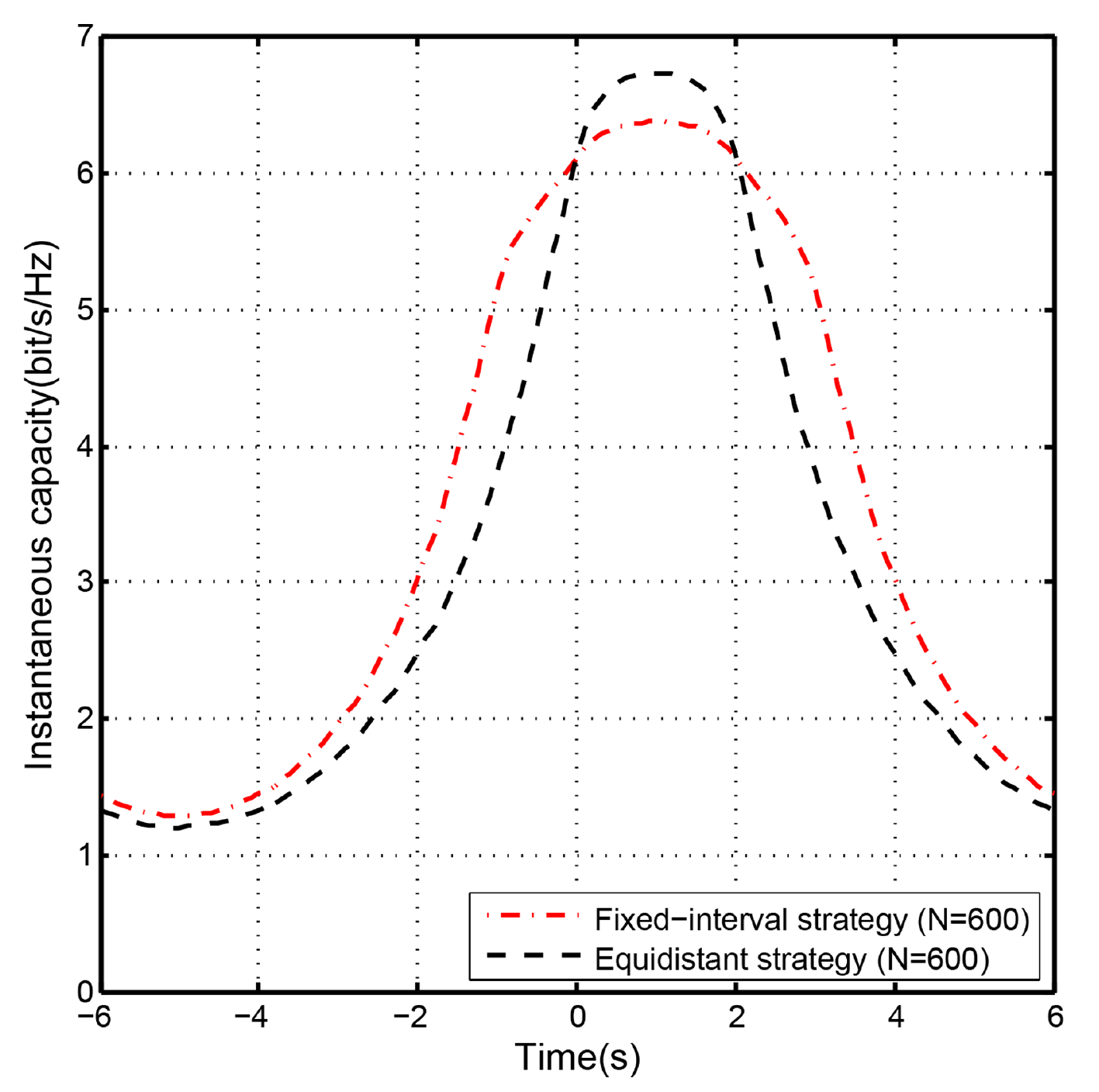}\picspace
\caption{The instantaneous channel capacity of fixed-interval strategy ($\delta=0.15m$) and equidistant strategy for $N=600$ case.}
\label{A600}
\end{figure}\picspace

For the ultimate case, where $N$ is increased to be $N_{\textrm{max}}$, Figure \ref{com_f_e} plots its service amount w.r.t the numbers of antennas of two strategies when the fix interval $\delta$ is set as 0.075m which is the half wavelength of the carrier. In such an extreme scenario, $N_{\textrm{max}}=2666$. It can be seen that, the two considered strategies have the overlapped performance at $N=2$ and $N=2666$, which is as expected. Moreover, it can be seen that the service amount of two proposed strategies is firstly increasing and then decreasing versus the number of antennas and the maximum difference comes about at $N=200$ in this scenario. The reason is that for a relatively large $N$, the average distance of two antennas in equidistant strategy is  similar to that of that fixed-interval strategy. Therefore, when for a relatively large $N$, with the increment of $N$, the system performance difference between the two strategies is decreased.

\begin{figure}
\centering
\includegraphics[width=0.25\textwidth]{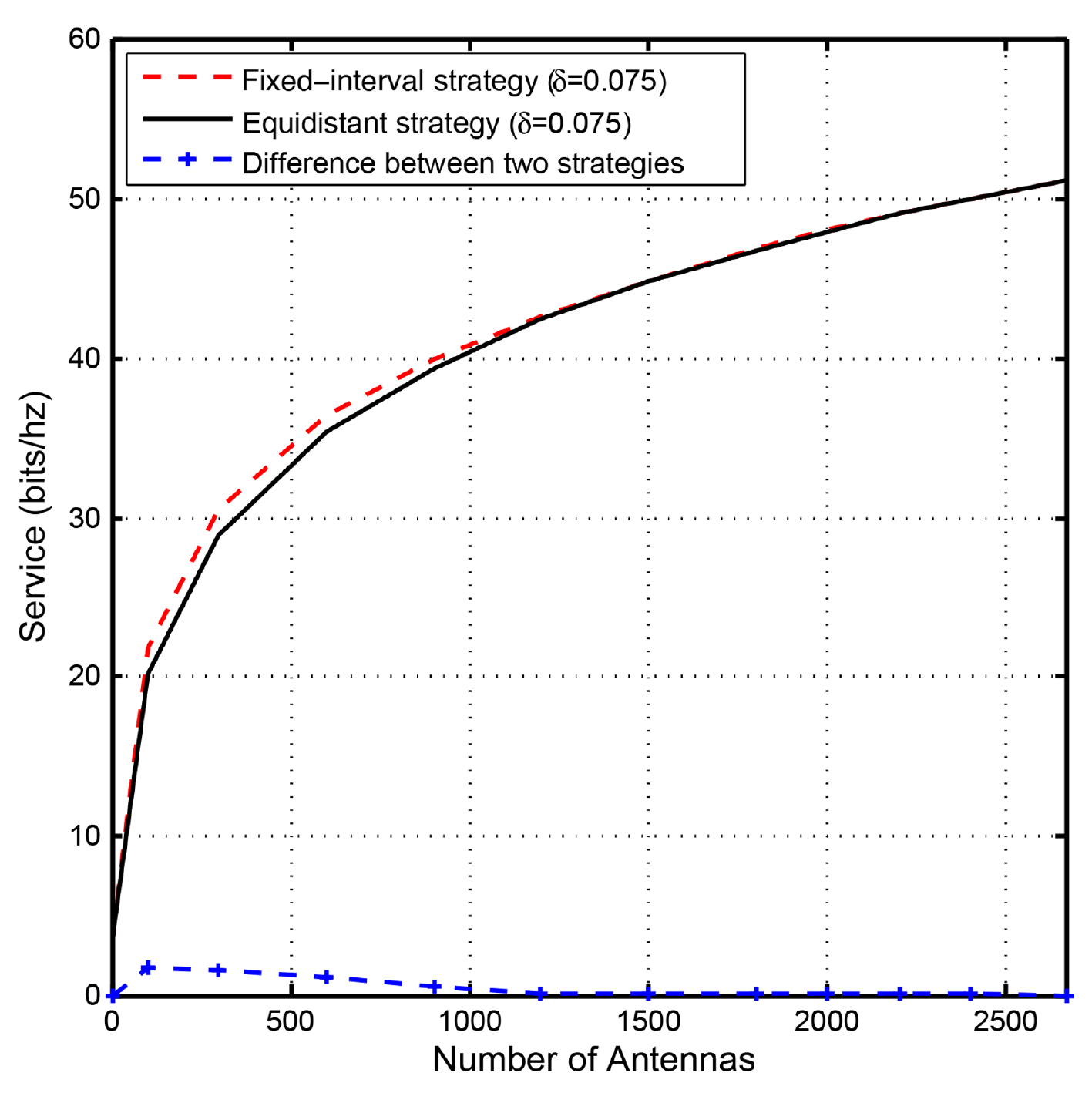}\picspace
\caption{The service amount of two strategies versus the number of antennas ($\delta=0.075m$). }
\label{com_f_e}
\end{figure}\picspace

To fully compare the two strategies, Figure \ref{end} plots the service and outage time ratio of the two strategies, where $\delta=1m$. The result shows that with the increment of the number of antennas, the service of both strategies increases and outage time ratio of both strategies decreases. Nevertheless,  the fixed-interval strategy show better performance than the equidistant strategy.
Based on the numerical results, we presented our conjecture above for the antenna number is relatively large.

\begin{figure}
\centering
\includegraphics[width=0.5\textwidth]{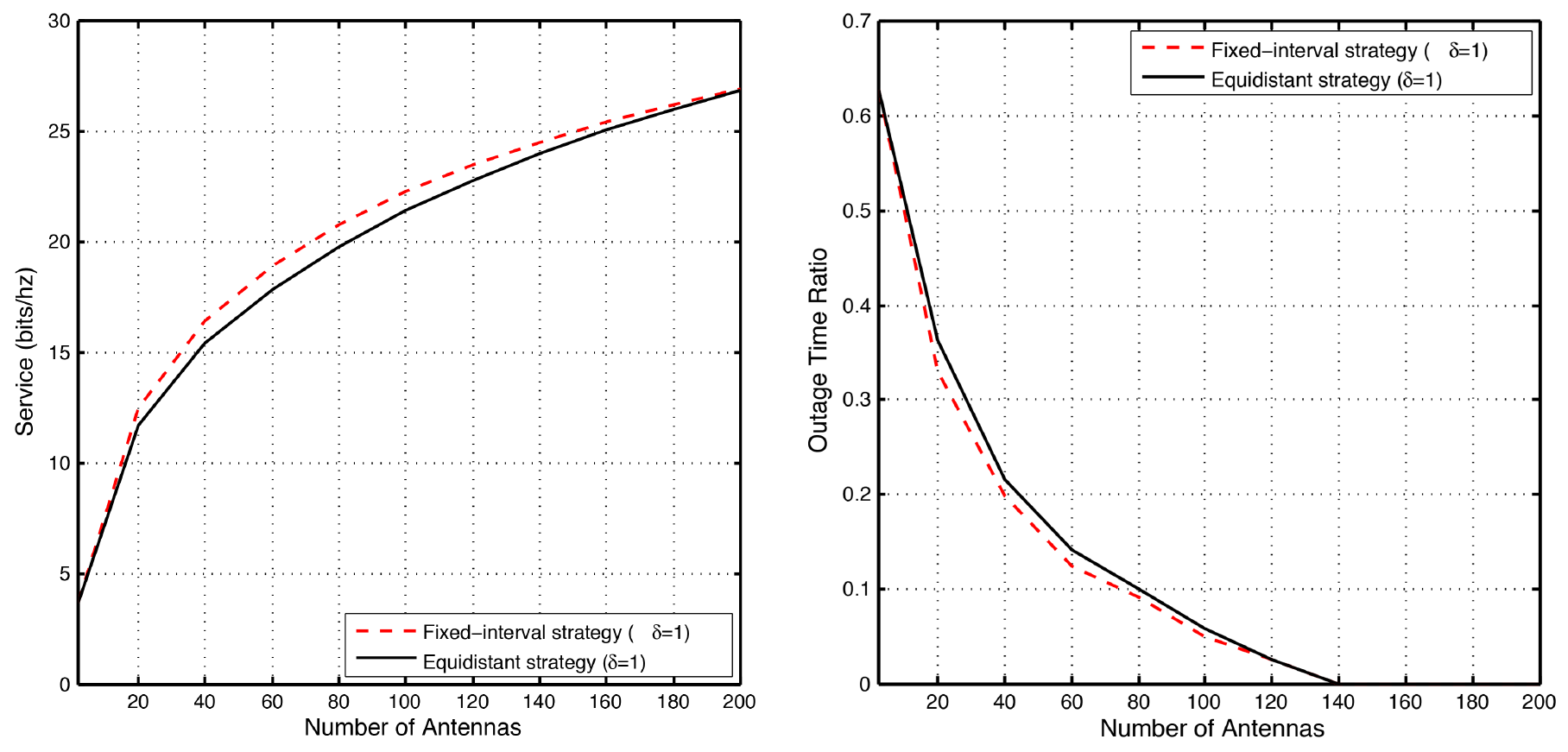}\picspace
\caption{The comparison between two proposed strategies based on (i) service and (ii) outage time ratio ($\delta=1m$).}
\label{end}
\end{figure}\picspace

Figure \ref{fix_com} plots the service amount and  outage time ratio of fixed-interval strategy with 20 antennas versus  $\delta$. The service amount decreases and the outage time ratio increases as $\delta$ increasing. This is because when $\delta$ increases the distance between the two sets of antennas in the fixed-interval strategy decreases. Therefore, a smaller $\delta$ is better for fixed-interval strategy. Thus, it needs to take care it in practice.

\begin{figure}
\centering
\includegraphics[width=0.5\textwidth]{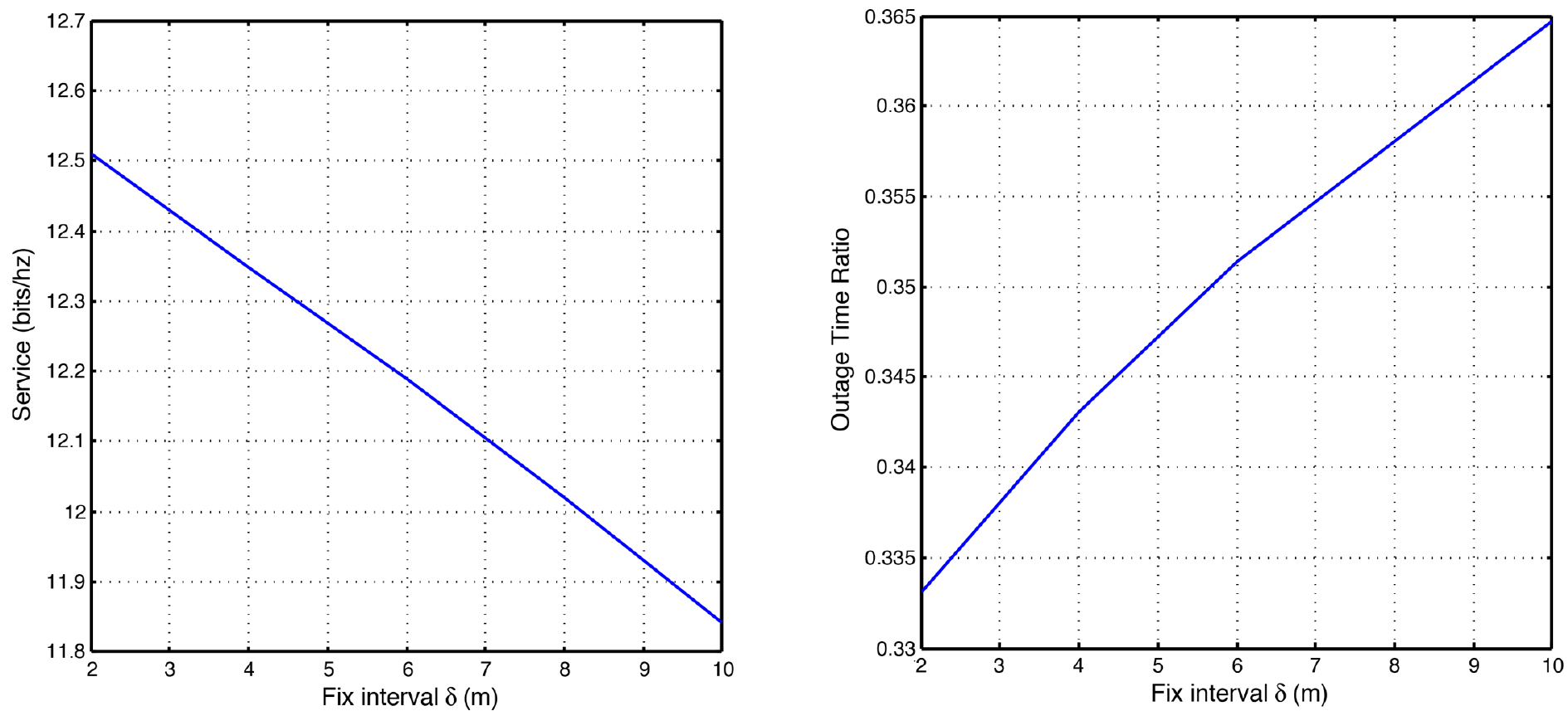}\picspace
\caption{The service amount and outage time ratio of fixed-interval strategy with different $\delta$ for $N=20$ case.}
\label{fix_com}
\end{figure}\picspace

\section{Conclusion}

This paper considered two multi-antenna deployment strategies, i.e., the equidistant strategy and the fixed-interval strategy. We discussed and compared the two strategies in terms of instantaneous channel capacity, service amount and outage time ratio. Numerical results show that, for two-antenna case, by increasing the distance of neighboring antennas in a reasonable region, the system performance can be enhanced and it is also shown that the two strategies have much difference performance behavior in terms of instantaneous channel capacity and the fixed-interval strategy  may achieve better performance that the equidistant one in terms of service amount and outage time ratio for the antenna number is much large.

%


%

\end{document}